\documentclass[11pt]{article} 
\usepackage[dvipsnames,svgnames,x11names]{xcolor}
\usepackage[utf8]{inputenc}
\usepackage{authblk}
\usepackage{fullpage}
\usepackage{tabularx}
\usepackage{amsthm}
\usepackage{amsmath}
\usepackage{amssymb}
\usepackage{amsfonts}
\usepackage{mathtools}
\usepackage{enumitem}
\usepackage[pagebackref]{hyperref}
\usepackage[sort,numbers]{natbib}
\usepackage{tikz}
\usepackage[capitalize]{cleveref}
\usepackage{todonotes}

\usepackage{multirow}
\usepackage{booktabs}
\usepackage{makecell}
\usepackage{caption}
\usepackage{comment}

\usepackage{tikz}

\usetikzlibrary{arrows,decorations.pathmorphing,decorations.pathreplacing,backgrounds,positioning,fit,matrix}
\usetikzlibrary{shapes,calc,patterns,arrows.meta}
\tikzset{
	vert/.style={circle,inner sep=1.5,fill=white,draw=black,minimum size=.3cm},
    dummy/.style={circle,fill=black,draw=black,inner sep=2.5},
	edge/.style={color=black, thick},
	diredge/.style={->,>={Stealth[width=8pt,length=8pt]},color=black, thick},
	timelabel/.style={fill=white,font=\footnotesize, text centered},
	wave/.style={decorate,decoration={coil,aspect=0}},
	dirwave/.style={->, >={Stealth[width=8pt,length=8pt]},decorate,decoration={coil,aspect=0}},
	diredge2/.style={->,>={Stealth[width=8pt,length=8pt]}}
}

\usetikzlibrary{shapes.geometric, arrows}

\tikzstyle{process} = [rectangle, 
minimum width=4cm, 
minimum height=1.5cm, 
text centered, 
text width=4cm, 
draw=black, 
fill=orange!30]

\tikzstyle{decision} = [diamond, 
minimum width=4.5cm, 
minimum height=1cm, 
aspect=3,
text centered,
text width=4cm,  
draw=black, 
fill=green!30]
\tikzstyle{arrow} = [thick,->,>=stealth]

\newtheorem{theorem}{Theorem}
\newtheorem{lemma}[theorem]{Lemma}
\newtheorem{observation}[theorem]{Observation}
\newtheorem{corollary}[theorem]{Corollary}

\newtheorem{definition}{Definition}

\theoremstyle{definition}

\usepackage{soul}

\crefname{figure}{Figure}{Figures}

\newcommand{\commentout}[1]{}

\newcommand{\problemdefopt}[3]{
	\begin{center}\fbox{
	\begin{minipage}{0.95\textwidth}
		\noindent
		#1
		\vspace{5pt}\\
		\setlength{\tabcolsep}{3pt}
		\begin{tabularx}{\textwidth}{@{}lX@{}}
			\textrm{Input:}     & #2 \\
			\textrm{Task:}  & #3
		\end{tabularx}
	\end{minipage}}
	\end{center}
}

\newcommand{\tlabel}[1]{#1\mid f(#1)\mid e(#1)}
\newcommand{\tttlabel}[3]{#1\mid #2\mid #3}

\begin{document}

\title{Tournament Robustness via Redundancy\footnote{This work was supported by the Israel Science Foundation, grant nr.~1470/24, by the European Union's Horizon Europe research and innovation programme under grant agreement~949707, and by the European Research Council, grant nr.~101039913 (PARAPATH).}}

\author{Klim~Efremenko}
\author{Hendrik~Molter}
\author{Meirav~Zehavi}

\affil{\small Department of Computer Science, Ben-Gurion~University~of~the~Negev, 
Beer-Sheva, 
Israel\\ \texttt{klim@bgu.ac.il, molterh@post.bgu.ac.il, meiravze@bgu.ac.il}}

\date{}

\maketitle
\begin{abstract}
A \emph{knockout tournament} is one of the most simple and popular forms of competition. Here, we are a given binary tournament tree where all leaves are labeled with seed position names. The players participating in the tournament are assigned to the seed positions. In each round, the two players assigned to leaves of the tournament tree with a common parent compete, and the winner is promoted to the parent. The last remaining player is the winner of the tournament.

In this work, we study the problem of making knock-out tournaments \emph{robust against manipulation}, where the form of manipulation we consider is changing the outcome of a game. We assume that our input is only the number of players that compete in the tournament, and the number of manipulations against which the tournament should be robust. Furthermore, we assume that there is a strongest player, that is, a player that beats any of the other players. However, the identity of this player is not part of the problem input.

To ensure robustness against manipulation, we uncover an unexpected connection between the problem at hand and communication protocols that utilize a feedback channel, offering resilience against adversarial noise. We explore the trade-off between the size of the robust tournament tree and the degree of protection against manipulation. Specifically, we demonstrate that it is possible to tolerate up to a $1/3$ fraction of manipulations along each leaf-to-root path, at the cost of only a polynomial blow-up in the tournament size.
\end{abstract}

\section{Introduction}

A \emph{knockout tournament}, also referred to as a single-elimination or cup tournament, is a very popular competition format in which participants are paired-up into matches in successive rounds, with the losers being eliminated after each round. This format is commonly used in sports~\cite{chaudhary2024make,CR11,GMSS12,suksompong2021tournaments,williams_moulin_2016} and also applies to other areas like elections and decision-making processes~\cite{vu2009complexity,laslier1997tournament,brandt2007pagerank,Tullock80,rosen1985prizes}. Formally, a knockout tournament involves~$n$ players, where $n$, w.l.o.g., is a power of $2$. The players are assigned to the leaves of a complete binary tree through a bijective mapping known as a {\em seeding}. As long as at least two players remain, those paired in the tree (i.e., those who share a common parent) compete in a match, with the winner, now mapped to the common parent, advancing to the next round. The losers are eliminated, and their corresponding leaves are removed from the tree. This process continues until only one player remains, who is declared the {\em winner}.

Knockout tournaments, due to their inherently competitive nature, are particularly susceptible to various forms of manipulation. These manipulations are frequently highlighted by the media~(e.g., \cite{news,news1,news2,manoli2015only,hill2010critical,feltes2013match}), and their vulnerability has been demonstrated empirically~(e.g., \cite{stanton2013structure,mattei2016empirical}). Most research on manipulation in knockout tournaments focuses on constructive manipulation, where the goal is to ensure a favored player wins the competition. The most widely studied forms of such manipulation include tournament fixing, bribery, coalition manipulation, and their combinations. The setting often assumes that the manipulator(s) possesses information on the strengths of the players, and, particularly, for every possible match between two players, who is more likely to win. Tournament fixing refers to recalculating a completely new seeding or altering an existing one~\cite{gupta2018rigging,gupta2019succinct,kim2017can,aziz2014fixing,zehavi2023tournament,gupta2018winning,vu2009complexity,williams2010fixing}. Bribery, on the other hand, refers to influencing the results of a limited or budget-constrained number of matches~\cite{konicki2019bribery,mattei2015complexity,kim2015fixing}.
Lastly, in coalition manipulation, the coalition players can intentionally lose any of their games~\cite{russell2009manipulating,mattei2015complexity,aaai25}. 
In this work, we adopt the opposite perspective---we position ourselves as the (decent) moderator, rather than the (corrupted) manipulator.

%\bigskip
%\noindent{\bf Robustness via Redundancy.}
\subsection{Robustness via Redundancy}
Our goal is to design knockout tournaments that are \emph{robust} against manipulation, with a specific focus on bribery as the form of manipulation. In particular, we aim to protect the tournament outcomes from bribery that could alter the results of a limited number of games (denoted by $k$) on any leaf-to-root path.\footnote{
Thus, in principle, we can manage a total number of manipulations much larger than 
$k$, as long as the number of manipulated games along each individual leaf-to-root path is limited to 
$k$.}
Our approach is to enlarge the tournament tree and allow multiple leaves of the tournament tree to be associated with the same player. This introduces redundancy, ensuring that the ``strongest player'' still wins the enlarged tournament, even in the presence of some manipulation. We explore the trade-off between the extent of the tournament's expansion and its level of robustness against manipulation.

By ``strongest player'', we mean a player who can defeat every other player participating in the tournament. It is important to emphasize that {\em we do not know who this player is in advance}. Our input is simply the number of players, without any information about who can beat whom, and, in particular, without knowing who the strongest player is (or even if one exists).
We choose this definition of the desired winner for three key reasons. First, a definition of a ``desired winner'' is essential (one that reflects fairness), and this is the simplest, making it the most suitable for a first study on robustness in tournament design, where the goal is to establish the foundational framework for this new area of research. Second, this definition does not rely on any knowledge of intricate estimations of which player can beat which other player (as some previous studies on bribery in tournaments do), which is very difficult---if not impossible---to determine in practice. Third, we believe the existence of a strongest player is reasonable in many contexts. For example, in a chess tournament, it is often the case that a single player can defeat all others, even though we may not know who this player is before the tournament starts (especially when a ``new star" rises to beat the reigning champion).

%\bigskip
%\noindent{\bf Main Technical Contribution.}
\subsection{Main Technical Contribution}
Deferring formal definitions to \cref{sec:model}, our main technical contribution can be summarized as follows.  (For the formal statement, see \cref{thm:binary}.)

\begin{theorem}[{\bf Informal}]\label{theorem:main}
Given a number of $n$ players and a non-negative integer $k$, there exists a complete binary tree $T$ with labeled leaves, such that for every set of $N$ players and every subjective mapping $\sigma$ from the players to the labels of the leaves the following holds:
\begin{itemize}
\item {\bf [Robustness.]} If a strongest player exists and the number of manipulated games on each leaf-to-root path is at most $k$, then this strongest player will emerge as the winner of the knockout tournament associated with $T$ and $\sigma$.
\item {\bf [Redundancy.]} The height of $T$ is at most $\lceil \frac{1}{\log_2((1+\sqrt{5})/2)}\cdot\log_2 n \rceil+3k \leq \lceil1.44\log_2 n\rceil+3k$.
\end{itemize}

Moreover, $T$ can be constructed in time linear in its size.
\end{theorem}

%\[
%f(\varepsilon)k+(1+\varepsilon)\log_2 n
%\]

The increase in the height of the robust tournament tree consists only of an additively linear factor in $k$ and an additive linear factor in the height of a ``standard tournament'' (being $\log_2 n$). 
This means that instead of having roughly $2n$ games in total, we have roughly $2\cdot (2^{3k}+n^{1.44})$ games in total. 
We focus in our paper on the case where $k$ is the magnitude of $\log_2 n$. This is specifically interesting because that means we can support a \emph{constant fraction} of manipulations on every path from a leaf to the root of the tournament tree. Note that our result implies that e.g.\ if we have $k=\log_2 n$, then we get a fraction of $1/4.44$ games that may be manipulated; and for larger $k$, this fraction tend to $1/3$.
In turn, this implies that the size of the tournament will be polynomially larger than than a ``standard'' tournament where the tournament tree has depth $\log_2 n$.

%is to be exponentially (or, at least, significantly) larger than $k$, it is important that the dependency of the redundancy on $k$ is additive, therefore ``separated'' from the dependency of the redundancy on $n$.

Additionally, before proving \cref{thm:binary}, we first prove the following result, which may be of independent interest (see \cref{thm:ternary}): If the constructed tournament allows matches involving three players instead of just two, then the height of the tree can be bounded by $\log_2n+2k$. The proof of this result is simpler than that of \cref{thm:binary}, yet it effectively illustrates the core ideas. 

It is important to note that the exponential dependence on 
$k$ in the total number of games is \textit{unavoidable}, as a linear dependence on 
$k$ in the tree height is also inevitable. In fact, if the height were less than 
$k$, then \textit{all} games in the tournament could be manipulated, making it possible to arbitrarily determine the winner.

%\bigskip
%\noindent{\bf Comparison with Other Methods to Ensure Robustness.}
\subsection{Comparison with Other Methods to Ensure Robustness}
Next, the following discussion is in order: we aim to compare our approach with two alternative methods that might naturally arise in the context of robustness. The most straightforward way to protect against 
$k$ manipulations is to repeat each game 
$2k+1$ times and declare the player who wins the majority as the match winner. Most critical, this alters the competition format. Essentially, we redefine what it means to play a game (e.g.\ best-of-$2k+1$ instead of best-of-one) within the tournament rather than adapting the tournament itself.
Hence, this is a conceptually different approach.
Moreover, the biggest drawback of this approach is that we cannot support a constant fraction of manipulated games in any tournament run, since the number of games in each run is $(2k+1)\log_2 n$. (Indeed, we have that $k/(2k+1)\log_2 n$ is proportional to $1/\log_2 n$ and therefore not a constant fraction.) 
%Also, arguable, this approach has the drawback that a manipulator might be able to change the outcome of a batch of games at once. 
%\todo{expand on this} 
%Another drawback is that it leads to roughly $2\cdot(2k+1)n\sim 4kn$ games in total, meaning the overhead no longer depends additively on $k$. 
%Intuitively, $k$ is expected to be significantly smaller than 
%$\log_2n$, as $\log_2n$ represents the full height of the tree from a leaf to the root. Consequently, redundancy of the form $\Theta(kn)$ seems considerably worse than the redundancy we achieve, which is of the form $\Theta(2^{\Theta(k)}+n^{1.44})$. Indeed, when $k$ is sufficiently smaller than $\log_2n$, our redundancy is simply $\Theta(n)$.

Another possible way to protect against \( k \) manipulations is to use a \((k+1)\)-elimination generalization of a knockout tournament. Informally, this involves running a standard knockout tournament, where the losers of each round move into a second identical tournament. The losers of the second tournament then enter a third one, and this process continues until the \((k+1)\)-th tournament. 
The winner of the \((k+1)\)-th tournament advances to compete against the winner of the \( k \)-th tournament, and the winner of that match moves on to face the champion of the \( (k-1) \)-th tournament. This progression continues until the final match, where the overall tournament winner is decided. This setup ensures that the strongest player cannot be eliminated before entering \((k+1)\)-th tournament. However, as is, it does not protect against \( k \) manipulations, as one could still rig the final game to ensure the strongest player loses.
To address this, the final game must be repeated \( 2k+1 \) times, with the winner determined by majority---similar to the previous approach. Likewise, the game before the final must be played \( 2(k-1)+1 \) times, and this pattern continues throughout the tournaments. However, this method is very cumbersome and structurally inelegant.
%Additionally, it exhibits the same drawback as the previously, that is, that we cannot support a constant fraction of manipulations in every tournament run.
%but also inefficient, as it results in a total of \( \Theta(k \cdot n) \) games, which, again, is non-additive.

%\bigskip
%\noindent{\bf Surprising Connection to Interactive Communication.} 
\subsection{Surprising Connection to Communication Protocols}
Our paper is inspired by Karchmer Wigderson games (KW-games)~\cite{karchmer1990monotone}, which show a correspondence between Boolean circuits and communication protocols.
Specifically, to prove \cref{thm:ternary,thm:binary}, we construct tournament trees based on a protocol for communication with a feedback channel that are resilient against adversarial noise (see~\cite{Berlekamp1964BlockCW}). 
%We present a construction of a tournament that is based on a result in interactive communication with a feedback channel and adversarial noise (see~\cite{Berlekamp1964BlockCW}).
%In turn, such translations between tree-like structures and protocols was first considered by Karchmer and Widgerson~\cite{karchmer1990monotone}.
We consider the setting where Alice aims to transmit a binary message to Bob over a communication system consisting of a transmission channel and a feedback channel. Specifically, whenever Alice sends a bit through the transmission channel, the feedback channel relays the bit received by Bob back to Alice.
Adversarial noise in this model allows a third party, Eve, to corrupt up to \( k \) bits of her choosing. 

We present a non-trivial translation of a communication protocol---one that guarantees Bob receives Alice’s message correctly (without errors)---into a knockout tournament format. In this translation, players are represented as binary strings, corresponding to their respective leaf-to-root paths in the tournament tree, while manipulated games are treated as errors introduced by Eve. We then prove that the constructed tournaments satisfy both properties stated in Theorem~\ref{theorem:main}.

This innovative link between tournament structures and communication protocols could be of independent interest, potentially inspiring future applications in tournament design and beyond.

\section{Tournaments and Manipulation}\label{sec:model}

In this section, we introduce the formal definitions of tournaments, tournament trees, manipulation, robustness, and all related concepts that we use in this paper. We remark that from now on, we will use ``$\log$'' to refer to the logarithm with base two.
%TODO: Add a discussion that we can allow the tournament to be not complete (i.e. not all leaves at same depth) without loss of generality, since we can always let a player play against itself, and the meaning of that is just that the player automatically advances).\todo{!}

%TODO: Remove all mentioning of a tournament graph from everywhere, and also remove all mentioning of the initial seeding $\sigma$.\todo{!}

\subsection{Basic Knock-Out Tournaments}
We use the following model for knock-out tournaments. 
A \emph{tournament tree} $T$ is a rooted full binary tree\footnote{A full binary tree is a binary tree where every inner vertex has either zero or two children.} where each leaf $i$ is associated with a bit string $s_i\in\{0,1\}^\ast$. We call that bit string the \emph{seed position names} of the tournament tree $T$. We denote with $S_T$ the set of seed positions of $T$.
We remark that we allow different leaves of $T$ to have the same seed position name. We call a tournament tree $T$ \emph{balanced} if $T$ is a rooted complete binary tree, that is, a full binary tree where all leaves have the same distance to the root. Note that we can turn any tournament tree into a balanced tournament tree by replacing leaves that are too close to the root with balanced binary trees of appropriate length where every leaf of the replacement tree has the same seed position name as the leaf of the original tournament tree that is replaced.
A tournament $\mathcal{T}$ consists of the following.
\begin{itemize}
    \item A set of $n$ \emph{players} $N=\{1,\ldots, n\}$. %, where for simplicity, we assume that $n$ is a power of two, that is, $n=2^r$. 
    \item A tournament tree $T$ with $n$ different seed positions.
    \item A \emph{winner function} $w:N\times N\rightarrow N$ that determines the winner of a game between two players in $N$.
    Formally, a function $w$ from $N\times N$ is a \emph{winner function} if for all $(i,j)\in N\times N$ we have that $w(i,j)\in\{i,j\}$. We remark that we allow players to play against themselves, that is, for all $i\in N$ we have that $w(i,i)=i$.
    %\todo{players can play against themselves. comment on this.}
    %\item A \emph{tournament graph} $D=(N,A)$, that is, a directed graph where for every pair of vertices an arc in exactly one of the two directions is contained in $A$.
    \item A \emph{tournament seeding} $\sigma:N\rightarrow S_T$, which is a bijective mapping between the player names and the seed position names of $T$.
\end{itemize}

The tournament $\mathcal{T}=(N,T,w,\sigma)$ is conducted in rounds as follows. 
Initially, each player $i$ is assigned to all leaves with seed position $\sigma(i)$.
As long as the generalized tournament tree $T$ has at least two leaves, every two players $i,j\in N$ that are assigned to leaves with a common parent in the tree $T$ plays against each other. 
%If $(i,j)\in A$ then player $i$ beats player $j$ with probability one. 
The winner $w(i,j)$ is promoted to the common parent. Then, the leaves are deleted from the tree $T$. Eventually, only one player remains, and this player is declared the winner of $\mathcal{T}$.

\subsection{Ternary Knock-Out Tournaments}
%We begin by defining the following model. 
A \emph{ternary tournament tree} $T$ is a rooted tree where each leaf $i$ is associated with a bit string $s_i\in\{0,1\}^\ast$, the root has degree three and each inner vertex (except the root) has degree four. This is the main difference to the basic setting. The remaining concepts are analogous.
%The labels of the leaves of $T$ are called \emph{seed positions} of the tournament tree $T$. We denote with $S_T$ the set of seed positions of $T$.
%We remark that we allow different leaves of $T$ to have the same seed position.
A generalized tournament $\mathcal{T}$ consists of the following.
\begin{itemize}
    \item A set of $n$ \emph{players} $N=\{1,\ldots, n\}$. %, where for simplicity, we assume that $n$ is a power of two, that is, $n=2^r$. 
    \item A generalized tournament tree $T$ with $n$ different seed positions.
    %\item A \emph{tournament graph} $D=(N,A)$, that is, a directed graph where for every pair of vertices an arc in exactly one of the two directions is contained in $A$.
    \item A \emph{winner function} $w: N\times N\times N\rightarrow N$ that determines the winner of a game between three players of players in $N$.
    Similar to the basic setting, a function $w$ from $N\times N\times N$ is a \emph{winner function} if for all $(i_1,i_2,i_3)\in N\times N$ we have that $w(i_1,i_2,i_3)\in\{i_1,i_2,i_3\}$. Again, we remark that we do not require $i_1$, $i_2$, and $i_3$ to be distinct.
    \item A \emph{tournament seeding} $\sigma:N\rightarrow S_T$, which is a bijective mapping between the player names and the seed positions of $T$.
\end{itemize}

The ternary tournament $\mathcal{T}=(N,T,w,\sigma)$ is conducted in rounds as follows. 
Initially, each player $i$ is assigned to all leaves with seed position $\sigma(i)$.
As long as the ternary tournament tree $T$ has at least three leaves, every three players $i_1,i_2,i_3$ that are assigned to leaves with a common parent in the tree $T$ play against each other. 
%If $(i,j)\in A$ then player $i$ beats player $j$ with probability one. 
The winner $w(i_1,i_2,i_3)$ is promoted to the common parent. Then, the leaves are deleted from the tree $T$. Eventually, only one player remains, and this player is declared the winner of $\mathcal{T}$.

\subsection{Manipulation in Tournaments}
In this section, we formally define \emph{manipulations} in a (ternary) tournament.
The conduction of a tournament $\mathcal{T}=(N,T,w,\sigma)$ is manipulated if at least one game at a vertex in $T$ is manipulated.
We say that the game at vertex $v$ in  $T$ is \emph{manipulated} during the conduction of $\mathcal{T}$ if the following happens.
%\todo[inline]{This is very informal. Need a formal definition. First, define a manipulation at a node. Then, say the tournament is manipulation if it is manipulation in at least one node. We need to define a manipulation function, that tells, at each node, who will be the winner, and if it is not the strongest, then define it as a manipulation. Then, define the winner of the tournament similarly to Section 1.2.}
\begin{itemize}
    \item In some round $r$ during the conduction of $\mathcal{T}$, all children of $v$ are leaves. Let $i,j\in N$ be the players assigned to the leaves of $v$. 
    
    In the ternary case, let $i_1,i_2,i_3$ be the three players assigned to the leaves of $v$.
    \item In round $r+1$, player $i^\star\in \{i,j\}$ with $i^\star\neq w(i,j)$ is assigned to $v$. 
    
    In the ternary case, player $i^\star\in \{i_1,i_2,i_3\}$ with $i^\star\neq w(i_1,i_2,i_3)$ is assigned to $v$ in round $r+1$.
    %\item According to the tournament rules, player $p$ should be promoted to $v$.
    %\item There is a leaf of $v$ that is labeled with player $p'\neq p$ and $p'$ is promoted to $v$.
\end{itemize}

\subsection{Manipulation Robustness in Tournaments}
In this section, we formally define our problem of making a tournament robust when there exists a strongest player. 
Given a tournament tree $T$ we call a path from a leaf of $T$ to the root a \emph{tournament run}.
We say that a winner function $w$ \emph{exhibits a strongest player} if there exists $i_w\in N$ such that $i_w=w(i_w,j)=w(j,i_w)$ for all $j\in N$. In the ternary case, we say that a winner function $w$ \emph{exhibits a strongest player} if there exists $i_w\in N$ such that $i_w=w(i_1,i_2,i_3)$ for all $\{i_1,i_2,i_3\}\subseteq N$ with $i_w\in \{i_1,i_2,i_3\}$. We call $i_w$ the \emph{strongest player}.
We consider the following problem.

\problemdefopt{(Ternary) Tournament Robustness}{Integers $n$ and $k$.}{Find a (ternary) tournament tree $T$ with $n$ seed positions, such that the following holds:

For each set of players $N$ of size $n$, each winner function $w$ that exhibits a strongest player $i_w\in N$, and each tournament seeding $\sigma$ we have that $i_w$ wins $\mathcal{T}=(N,T,D,\sigma)$, even if the outcome of up to $k$ games in each tournament run in $T$ are manipulated when $\mathcal{T}$ is conducted.
}

%\problemdefopt{Simplified Tournament Robustness}{A player set $N$, a tournament seeding $\sigma$, and an integer $k$.}{Find a generalized tournament tree $T'$ of minimum height and a generalized tournament seeding~$\sigma'$, such that the following holds: 

%For each tournament graph $D$ we have that if player $i\in N$ wins $\mathcal{T}=(N,T,D,\sigma)$, where $T$ is the ordered rooted complete binary tree with $|N|$ leaves, then player $i$ also wins $\mathcal{T}'=(N,T',D,\sigma')$ even if the outcome of up to $k$ games in each tournament run in $\mathcal{T}'$ can be manipulated.}

\section{Communication with Feedback Channel and Adversarial Noise}
In this section, we introduce a simple communication protocol that is the main inspiration for our approach to make tournaments robust against manipulation.
We consider the setting where Alice wishes to transmit a (binary) message to Bob via a communication line consisting of a transmission channel and a feedback channel. This means that whenever Alice sends a bit of the message to Bob (via the transmission channel), the feedback channel sends the bit received by Bob back to Alice. We assume that Eve can manipulate the messages sent by Alice to Bob via the transmission channel, but Eve cannot manipulate the feedback channel. For an illustration see \cref{fig:comm}. This manipulation ability of Eve is also called \emph{adversarial noise}.

\begin{figure}[t]
\begin{center}
\includegraphics[scale=1]{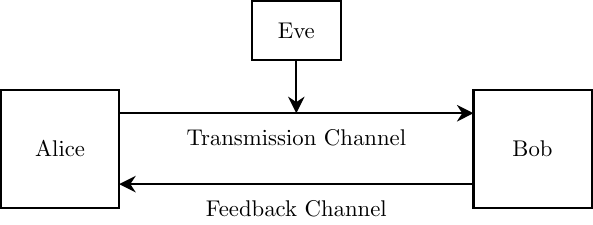}
%\begin{tikzpicture}[line width=1pt,scale=1]
%
%    \node[rectangle,draw,minimum width=1.5cm,minimum height=1cm] (e) at (0,1) {Eve};
%    \node[rectangle,draw,minimum width=2cm,minimum height=2cm] (a) at (-4,-1) {Alice};
%    \node[rectangle,draw,minimum width=2cm,minimum height=2cm] (b) at (4,-1) {Bob};
%    \draw[diredge] (0,.5) -- (0,-.4);
%    \draw[diredge] (-3,-.4) --node[label=below:Transmission Channel] {} (3,-.4);
%    \draw[diredge] (3,-1.6) --node[label=below:Feedback Channel] {} (-3,-1.6);
%\end{tikzpicture}
    \end{center}
    \caption{Illustration of the communication setting with feedback channel and adversarial noise.}\label{fig:comm}
\end{figure}

\subsection{Using an Extra Error Symbol}\label{sec:protocol1}

We recall a simple protocol that uses an additional symbol that Alice may use to indicate that the previously transmitted bit was manipulated. Assume that the alphabet of the message is~$\{0,1\}$ and the additional symbol is $\bot$. Intuitively, whenever Alice realizes (using the feedback channel) that the symbol received by Bob is not the one that she sent (that is, Eve manipulated the transmission), she sends a $\bot$ symbol to indicate to Bob that the last received symbol was manipulated. Bob then disregards the last received symbol and Alice repeats the transmission of the manipulated symbol. However, since Eve can also turn $\bot$ symbols into ones or zeros, Alice has to count how many manipulations still need to be recognized by Bob, and keep sending $\bot$ symbols until Bob has correctly identified all manipulations.

Formally, the protocol works as follows. For a visualization see \cref{fig:flow1}.
We assume that Alice has a manipulation counter that is initially set to zero. Alice uses this counter to keep track of how many manipulations still need to be recognized by Bob.
\begin{enumerate}
    \item Send the current bit (zero or one) of the message (starting with the first bit).\label{step:1}
    \item If the bit is correctly received by Bob, consider the next bit of the message to be the current bit, and move to step \ref{step:1}.
    \item If the bit is incorrectly received by Bob, then do the following:
    \begin{itemize}
        \item If Bob received the $\bot$ symbol instead of zero or one, consider the previous bit of the message to be the current bit (if there is no previous bit, the current bit stays the first bit), and move to step \ref{step:1}.
        \item If Bob received an incorrect bit (zero instead of one or vice versa), then send symbol $\bot$ to Bob and increase the manipulation counter by one.
    \end{itemize}
    \item If the symbol $\bot$ is correctly received by Bob, then decrease the manipulation counter by one.\label{step:4}
    \item If the symbol $\bot$ is incorrectly received by Bob, then increase the manipulation counter by one.
    \item If the manipulation counter is larger than zero, then send symbol $\bot$ to Bob and move to step~\ref{step:4}.
    \item If the manipulation counter is zero, then move to step \ref{step:1}.
\end{enumerate}

 It is folklore that Bob can reconstruct the original message as follows. Whenever a bit (zero or one) is received, append this bit to the message (starting with an empty message). Whenever an $\bot$ is received, delete the last bit from the message.

\begin{figure}
    \centering
    
\includegraphics[scale=1]{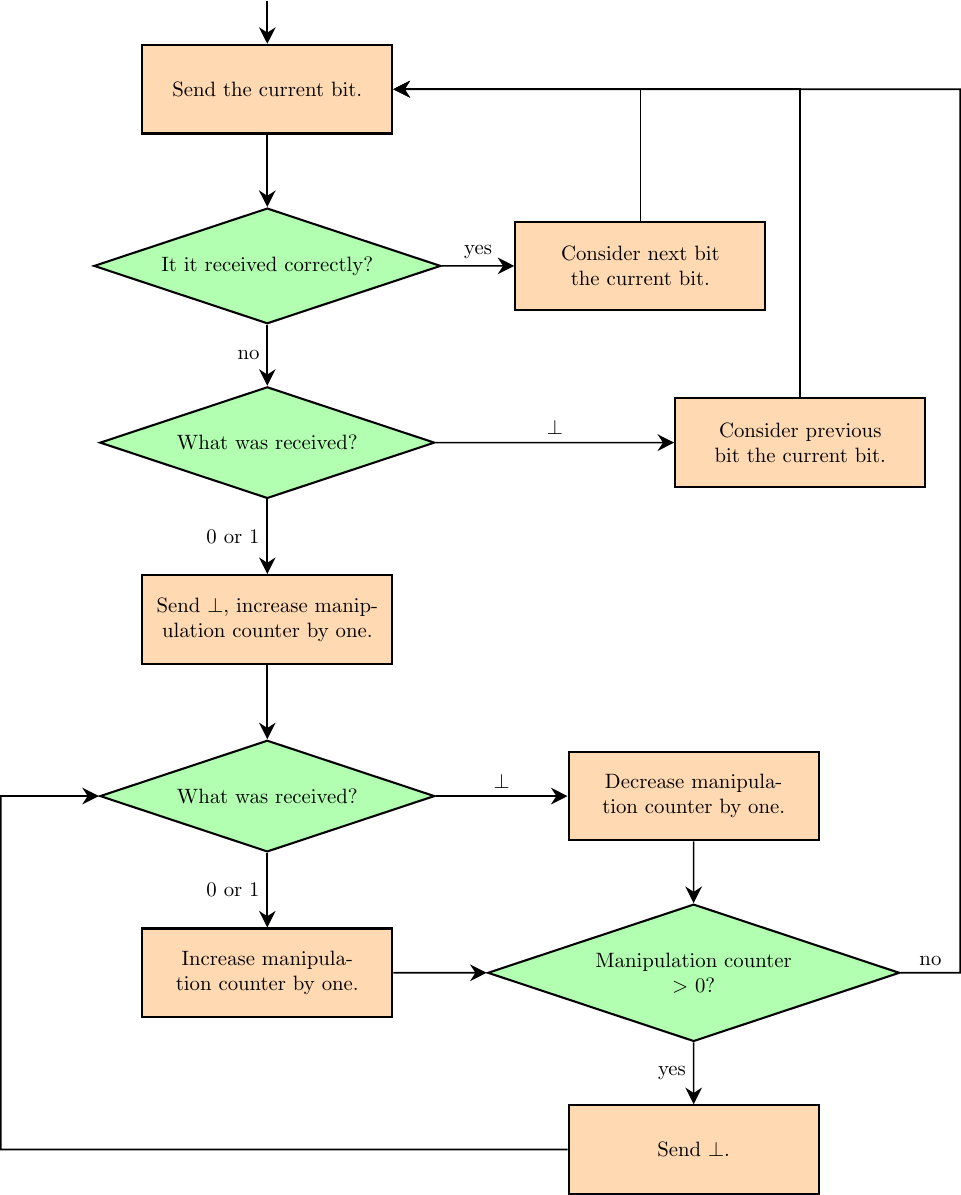}
    \caption{Flowchart of the protocol described in \cref{sec:protocol1}. Initially, the current bit is the first bit of the message and the manipulation counter is set to zero.}
    \label{fig:flow1}
\end{figure}

\subsection{Using a Binary Alphabet}\label{sec:protocol2}

We now recall a protocol that does not need an additional symbol. Intuitively, we replace the $\bot$ symbol from the previous protocol with the string 11. This requires the original message string to not contain any consecutive ones.
We remark that we could also use a longer sequence of ones to replace the $\bot$ symbol, which would allow the original message to include shorter sequences of ones. However this also requires Alice to send more symbols to indicate manipulations to Bob. Similarly to the first protocol, Alice has to keep track on how many manipulations still need to be recognized by Bob. However, since indicating a manipulation to Bob requires Alice to correctly transmit two symbols, the counter value does not directly correspond to the number of manipulations that Bob needs to recognize, but rather then number of consecutive ones that Bob needs to correctly receive.

The protocol works as follows. 
We assume that Alice has a manipulation counter that is initially set to zero.
\begin{enumerate}
    \item Send the current bit of the message (starting with the first bit).\label{step:11}
    \item If the bit is correctly received by Bob, consider the next bit of the message to be the current bit, and move to step \ref{step:11}.
    \item If the bit is incorrectly received by Bob, then do the following:
    \begin{itemize}
        \item If Bob received a one instead of a zero and the previous bit is a one, consider the previous previous bit of the message to be the current bit, and move to step \ref{step:11}.
        \item If Bob received a one instead of a zero and the previous bit is a zero, consider the previous bit of the message to be the current bit, send a one to Bob, and increase the manipulation counter by one.
        \item If Bob received a one instead of a zero and it was the first bit, send a one to Bob, and increase the manipulation counter by one.
        \item If Bob received a zero instead of a one, then send a one to Bob and increase the manipulation counter by two.
    \end{itemize}
    \item If the one is correctly received by Bob, then decrease the manipulation counter by one.\label{step:41}
    \item If the one is incorrectly received by Bob, then increase the manipulation counter by two.
    \item If the manipulation counter is larger than zero, then send a one to Bob and move to step \ref{step:41}.
    \item If the manipulation counter is zero, then move to step \ref{step:11}.
\end{enumerate}

This protocol was considered by Berlekamp~\cite{Berlekamp1964BlockCW}. Bob can reconstruct the original message as follows. Whenever a bit (zero or one) is received, append this bit to the message (starting with an empty message). If the message ends with two consecutive ones, then delete the last three bits from the message.

%\clearpage

\section{Tournament Robustness}

In this section, we first explain how the conduction of a tournament relates to communication protocols. Then we introduce a ternary robust tournament tree that is inspired by the communication protocol described in \cref{sec:protocol1}. Finally, we present a binary robust tournament tree based on the communication protocol described in \cref{sec:protocol2}.

\subsection{Tournaments and Communication Protocols}
Consider a binary tree where the leaves are labeled with binary strings and the inner vertices are labeled with the largest common prefix in the following way. The root is labeled with the empty string, and for each inner vertex, the left child adds a 0 to the label and the right child adds a 1 to the label of that inner vertex. For an example see \cref{fig:labeling}.

\begin{figure}[t]
%\noindent\makebox[\textwidth]{
\centering

\includegraphics[scale=1]{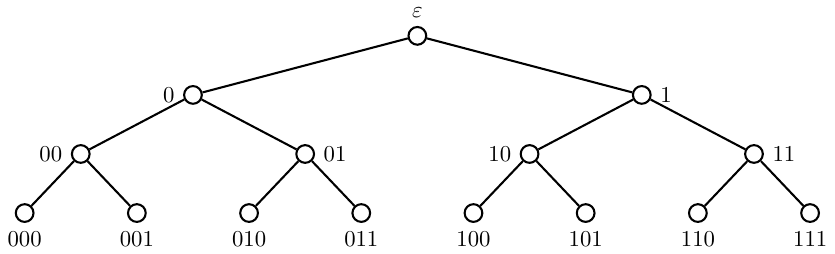}
%\begin{center}
%\begin{tikzpicture}[line width=1pt,yscale=1,xscale=1.9]
%
%    \node[vert,label=above:$\varepsilon$] (v0) at (2,1) {};
%
%    \node[vert,label=left:$0$] (v) at (0,0) {};
%    \node[vert,label=left:$00$] (v1) at (-1,-1) {};
%    \node[vert,label=right:$01$] (v2) at (1,-1) {};
%    \node[vert,label=below:$000$] (v21) at (-1.5,-2) {};
%    \node[vert,label=below:$001$] (v22) at (-.5,-2) {};
%    \node[vert,label=below:$010$] (v23) at (.5,-2) {};
%    \node[vert,label=below:$011$] (v24) at (1.5,-2) {};
%
%    \node[vert,label=right:$1$] (w) at (4,0) {};
%    \node[vert,label=left:$10$] (w1) at (3,-1) {};
%    \node[vert,label=right:$11$] (w2) at (5,-1) {};
%    \node[vert,label=below:$100$] (w21) at (2.5,-2) {};
%    \node[vert,label=below:$101$] (w22) at (3.5,-2) {};
%    \node[vert,label=below:$110$] (w23) at (4.5,-2) {};
%    \node[vert,label=below:$111$] (w24) at (5.5,-2) {};
%
%    \draw (v0) -- (v);
%    \draw (v0) -- (w);
%    
%    \draw (v) -- (v1);
%    \draw (v) -- (v2);
%    \draw (v1) -- (v21);
%    \draw (v1) -- (v22);
%    \draw (v2) -- (v23);
%    \draw (v2) -- (v24);
%      
%    \draw (w) -- (w1);
%    \draw (w) -- (w2);
%    \draw (w1) -- (w21);
%    \draw (w1) -- (w22);
%    \draw (w2) -- (w23);
%    \draw (w2) -- (w24);
%\end{tikzpicture}
 %   \end{center}
%    }
    \caption{Visualization of a binary tree where the root is labeled with an empty string $\varepsilon$ and the left child of every inner vertex adds a 0 to the label of its parent and the right child adds a 1 to the label of its parent.}\label{fig:labeling}
\end{figure}

We can interpret every path from the root to a leaf as a transmission of a bit string from Alice to Bob, where ``turning left'' means that a zero was transmitted, and ``turning right'' means that a one was transmitted.

Similarly, we can interpret the tournament run of the winning player as the transmission (in reverse order) of the label of its seed position. More specifically, the outcome of the final game corresponds to the transmission of the first bit. If the winner of the final game is the player from the left child, the first bit is a zero, and otherwise, it is a one. Then the semi-final game where the overall winner of the tournament participated corresponds to the transmission of the second bit, and so on.
A manipulation of a game then corresponds to adversarial noise, that is, Eve flipping the bit of the transmission. We can use ideas analogous to ones used in communication protocols that are robust to adversarial noise to make tournaments robust to manipulation of games. In the case where a third symbol is used in the communication protocol, the corresponding tournament will be ternary and the third child of an inner node of the tournament tree can be interpreted as transmitting the error symbol.

\subsection{The Ternary Case}

%\paragraph{Description of the robust tournament tree.}
In the following, we describe recursively how to obtain a tournament tree $T_{n,k}$ for a tournament with $n=2^r$ players. Here, $k$ denotes the number of manipulations that may occur in each tournament run. We will prove the following.

\begin{theorem}\label{thm:ternary}
Given integers $n$ and $k$, a tournament tree $T_{n,k}$ of height at most $\log n + 2k$ can be computed in $O(|T_{n,k}|)$ time such that the following holds.
Let $N$ be a set of $n$ players, let $w$ be a winner function that exhibits a strongest player $i_w\in N$, and let $\sigma$ be a seeding for the players to the seed position names of $T_{n,k}$.
    If on each tournament run in $\mathcal{T}=(N,T_{n,k},w,\sigma)$ we have that at most $k$ games are manipulated, then $i_w$ wins the tournament.
\end{theorem}

Assume we are given two integers $n$ and $k$. We start by describing how to construct the tournament tree $T_{n,k}$. We label each vertex of $T_{n,k}$ with three parts:
\[
\tlabel{s},
\]
where $s$ is a string on the alphabet $\{0,1,\bot\}$, $k$ is a non-negative integer, $f:\{0,1,\bot\}^*\rightarrow\{0,1\}^*$ is a function mapping strings on alphabet $\{0,1,\bot\}$ to binary strings, and $e:\{0,1,\bot\}^*\rightarrow\mathbb{N}$ is a function mapping strings on the alphabet $\{0,1,\bot\}$ to non-negative integers. The functions $f$ and $e$ are defined as follows.
\begin{itemize}
    \item The function $f$ takes as input a string $s$ and exhaustively removes all substrings $0\bot$ and $1\bot$, as well as all leading $\bot$-symbols, from $s$.
    \item The function $e$ takes as input a string and counts how many times the function $f$ removes a substring of the form $0\bot$ or $1\bot$, or a leading $\bot$-symbol, from $s$.
    %, more formally
    %\[
    %e(s)=\frac{|s|-|f(s)|}{2}.
    %\]
\end{itemize}

Intuitively, the first part of the label corresponds to the transmission sequence as it is sent by Alice to Bob. The second part corresponds to the sequence that is reconstructed by Bob and believed by him to be the unmanipulated string. The third part counts how many manipulations were already identified by Bob.

We label the root of the tournament tree $T_{n,k}$ with label $\tlabel{\varepsilon}$, where $\varepsilon$ is the empty string. The label simplifies to $\tttlabel{\varepsilon}{\varepsilon}{0}$. A vertex with label $\tlabel{s}$ such that $|f(s)|\le r$ and $e(s) \le k$ has up to three children.
\begin{itemize}
    \item If $|f(s)|< r$, then the vertex has one child with label $\tlabel{s0}$.
    \item If $|f(s)|< r$, then the vertex has one child with label $\tlabel{s1}$.
    \item If $e(s)< k$ and $|f(s)|>0$, then the vertex has one child with label $\tlabel{s\bot}$.
\end{itemize}
Here $s0$, $s1$, and $s\bot$ are the strings $s$ with a $0$, a $1$, and a $\bot$-symbol appended, respectively. We give an illustration in \cref{fig:ternary1}. 
Note that a vertex with label $\tlabel{s}$ such that $|f(s)| = r$ and $e(s)=k$ is a leaf.
Then we consider the bit string $f(s)$ to be the seed position name associated with the leaf.

\begin{figure}[t]
\begin{center}

\includegraphics[scale=1]{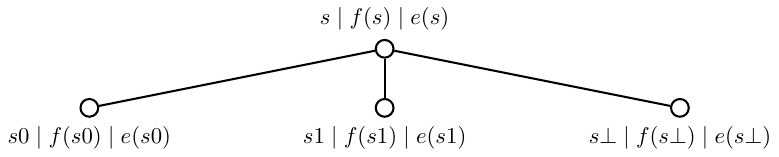}
%\begin{tikzpicture}[line width=1pt,scale=1,xscale=2.5]
%
%    \node[vert,label=above:$\tlabel{s}$] (v) at (0,0) {};
%    \node[vert,label=below:$\tlabel{s0}$] (v1) at (-2,-1) {};
%    \node[vert,label=below:$\tlabel{s1}$] (v2) at (0,-1) {};
%    \node[vert,label=below:$\tlabel{s\bot}$] (v3) at (2,-1) {};
%      
%    \draw (v) -- (v1);
%    \draw (v) -- (v2);
%    \draw (v) -- (v3);
%\end{tikzpicture}
    \end{center}
    \caption{Visualization of the recursive construction of $T_{n,k}$ in the ternary case.}\label{fig:ternary1}
\end{figure}

Note that given the label $\tlabel{s}$ of a vertex, we can compute the labels of the children in constant time. To do this, we do not recompute the functions $f$ and $e$ but rather compute the values of $f$ and $e$ in the labels of the children from the values in the label $\tlabel{s}$ in a straightforward manner. This allows us to observe the following.

%The intuition here is that $s$ corresponds to a string on the alphabet $\{0,1,\bot\}$ transmitted from Alice to Bob. The binary string $f(s)$ is the message recovered by Bob and the number $e(s)$ counts how many manipulations by Eve are reverted by Bob.

\begin{observation}\label{obs:timeternary}
    The tournament tree $T_{n,k}$ can be computed in $O(|T_{n,k}|)$ time.
\end{observation}

Now we analyze the height of $T_{n,k}$.

\begin{lemma}\label{lem:ternaryheight}
    The height of $T_{n,k}$ is $\log n + 2k$.
\end{lemma}
\begin{proof}
Recall that $n=2^r$.
    Let $v^\star$ be a vertex in $T_{n,k}$ that is labeled with $\tlabel{s}$. We associate with $\tlabel{s}$ the value $r-|f(s)|+2(k-e(s))$. By construction, the vertex $v^\star$ in $T_{n,k}$ has up to three children, labeled $\tlabel{s0}$, $\tlabel{s1}$, and $\tlabel{s\bot}$, respectively. We have that $\tlabel{s0}$ is associated with the value $r-|f(s0)|+2(k-e(s0))= r-|f(s)|+2(k-e(s))-1$. The situation with $\tlabel{s1}$ is analogous. Furthermore, we have that $\tlabel{s\bot}$ is associated with the value $r-|f(s\bot)|+2(k-e(s\bot))= r-|f(s)|+2(k-e(s))-1$. It follows that each child of a vertex is associated with a value that is by one smaller than that of the vertex itself. Clearly, the value cannot become negative. The root of $T_{n,k}$ is associated with value $r+2k$. The lemma statement follows.
\end{proof}

Now we show that for every set $N$ of $n$ players, every winner function $w$ that exhibits a strongest player $i_w$, and every seeding for the players to the seed position names of $T_{n,k}$, the tournament $\mathcal{T}=(N,T_{n,k},w,\sigma)$ is robust against $k$ manipulations in every tournament run. This means that $i_w$ wins $\mathcal{T}$ even if up to $k$ games in every tournament run are manipulated.
To this end, we compare the conduction of $\mathcal{T}$ to the conduction of a ``normal'' tournament with $n$ players.
Let $T_n$ denote the ordered rooted complete binary tree with $n$ leaves.
Now we label the root of $T_n$ with the empty string, and for each inner vertex, the left child adds a 0 to the label and the right child adds a 1 to the label. For an example see \cref{fig:labeling}.
Now every leaf of $T_n$ be labeled with a different bit string of length $\log n$ that serves as the seed position name of that leaf.
From the definitions of $T_n$ and $T_{n,k}$ it follows that both tournament trees have the same set of seed position names.
%\paragraph{Transforming a tournament to a robust one.}
%Assume we are given a player set $N$ with $|N|=n$, a tournament seeding $\sigma$, and an integer $k$. We construct a generalized tournament tree $T_{r,k}$ with $r=\log n$ as described above. Let $T$ be the ordered rooted complete binary tree with $|N|$ leaves. Assume the vertices of $T$ are labeled with binary strings in the following way. The root is labeled with the empty string, and for each inner vertex, the left child adds a 0 to the label and the right child adds a 1 to the label. For an example see \cref{fig:labeling}.
%We identify player $i$ with the binary string $s_i$ of length $r$ that is the label of the leaf $\sigma^{-1}(i)$ of $T$. We define a generalized tournament seeding $\sigma'$ as follows. We seed player $i$ into every leaf of $T_{r,k}$ that is labeled with $\tlabel{s}$ such that $f(s)=s_i$. It is straightforward to see that each player is seeded into at least one leaf of $T_{r,k}$.
%From \cref{lem:ternaryheight} and the fact that we set $r = \log n$ we get the following.
%\begin{corollary}\label{cor:ternaryheight}
 %   The height of the robust tournament tree $T_{n,k}$ is $\log n + 2k$.
%\end{corollary}

To show that $T_{n,k}$ is robust, we introduce some additional notation. 
\begin{definition}
Let $s$ be a binary string and let $\ell$ be some integer. If $|s|\ge \ell\ge 0$, then we denote with $p_\ell(s)$ the prefix of length $|s|-\ell$ of $s$, that is, the string obtained from $s$ by removing the last $\ell$ bits. If $\ell>|s|$, then $p_\ell(s)$ is the empty string.
\end{definition}
In other words, $p_\ell(s)$ is the prefix obtained from $s$ by removing the number of bits that correspond to detecting $\ell$ additional errors.
Intuitively, we will show that for each vertex $v^\star$ labeled with $\tlabel{s}$ in $T_{n,k}$ it holds that if at most $\ell$ manipulations already happened in the subtree below $v^\star$, then the strongest player wins the game at $v^\star$ if the strongest player wins the game at the vertex labeled with $p_{k-e(s)-\ell}(f(s))$ in $T_n$ without any manipulations. Informally speaking, the number $k-e(s)-\ell$ quantifies the number of manipulations that we can still support ``above'' vertex $v^\star$. The bit string $p_{k-e(s)-\ell}(f(s))$, intuitively, is the label of the vertex in $T_n$ to which the tournament tree $T_{n,k}$ ``simulates'' the advancement of the strongest player.
Formally, we prove that the following invariant holds during the conduction of $\mathcal{T}$.
%If $\ell<0$, then $p_\ell(s)$ can be any string obtained from $s$ by appending $-\ell$ bits to $s$.
\begin{lemma}\label{lem:ternarycorrect}
Let $N$ be a set of $n$ players, let $w$ be a winner function that exhibits a strongest player $i_w\in N$, and let $\sigma$ be a seeding for the players to the seed position names of $T_{n,k}$. 
    Let $s$ be a string on the alphabet $\{0,1,\bot\}$ such that $|f(s)|\le r$ and $e(s)\le k$, and let $0\le \ell\le k-e(s)$. If on each tournament run in $\mathcal{T}=(N,T_{n,k},w,\sigma)$ from a leaf of $T_{n,k}$ to the vertex $v^\star$ labeled with $\tlabel{s}$ we have that at most $\ell$ games are manipulated, then the following holds.
    %If the seed position name $\sigma(i_w)$ assigned to $i_w$ has $p_{k-e(s)-\ell}(f(s))$ as a prefix, then $i_w$ is the winner of $v^\star$.
    %
    If $i_w$ wins the game at the vertex labeled with $p_{k-e(s)-\ell}(f(s))$ of $T_n$ when the tournament $\mathcal{T}'=(N,T_n,w,\sigma)$ is conducted without any manipulations, then $i_w$ wins the game at $v^\star$.
\end{lemma}
\begin{proof}
    We prove the lemma by induction on the depth of $T_{n,k}$ (from the leaves to the root). Let $v$ be a leaf in $T_{n,k}$. By definition, $v$ is labeled with some $\tlabel{s}$ such that $e(s)=k$. Hence, we have $\ell=0$. By definition, we also have that the seed position name of leaf $v$ is $f(s)$. Since $p_0(f(s))=f(s)$, the statement holds.

    %We prove the lemma by induction on the depth of $T_{n,k}$ (from the leaves to the root). Each leaf $v$ in $T_{n,k}$ is labeled with some $\tlabel{s}$ such that $e(s)=k$. Hence, we have $\ell=0$. By definition, we have that the seeding for $T_{r,k}$ maps $v$ to the same player as the leaf labeled with $f(s)$ is mapped to in $T$. Since $p_0(f(s))=f(s)$, the statement holds.

    Consider a vertex $v^\star$ labeled with $\tlabel{s}$ such that $|f(s)|\le r$ and $e(s)\le k$ in $T_{n,k}$ that is not a leaf. Assume that on each tournament run from a leaf to this vertex, we have that at most $\ell\le k-e(s)$ games are manipulated. Consider the following cases.
    \begin{itemize}
        \item We have that $|f(s)|= r$. Then $e(s)<k$, otherwise $v^\star$ is a leaf in $T_{n,k}$. 
        
        By definition, $v^\star$ has one vertex as a child labeled with $\tlabel{s\bot}$. Note that $e(s)=e(s\bot)-1$. If the game at $v^\star$ is not manipulated, then by induction, the game at the child of~$v^\star$ is won by $i_w$ if the game at the vertex labeled with $p_{k-e(s\bot)-\ell}(f(s\bot))=p_{k-e(s)-\ell}(f(s))$ in $T_n$ is won by $i_w$ when $\mathcal{T}'$ is conducted without manipulations. 
        If the game at $v^\star$ is manipulated, then by induction, the game at the child of $v^\star$ is won by $i_w$ if the game at the vertex labeled with $p_{k-e(s\bot)-\ell+1}(f(s\bot))=p_{k-e(s)-\ell+1}(f(s))$ in $T_n$ is won by $i_w$ when $\mathcal{T}'$ is conducted without manipulations. 
        Note that if the game at the vertex labeled with $p_{k-e(s)-\ell}(f(s))$ in $T_n$ is won by $i_w$ when $\mathcal{T}'$ is conducted without manipulations, then $i_w$ also wins the game at the vertex labeled with $p_{k-e(s)-\ell+1}(f(s))$ in $T_n$, since the latter vertex is the parent of the former, and $i_w$ is the strongest player. 
        Hence, in both cases, we can conclude that the game at $v^\star$ is won by $i_w$ if the game at the vertex labeled with $p_{k-e(s)-\ell}(f(s))$ in $T_n$ is won by $i_w$ when $\mathcal{T}'$ is conducted without manipulations, and the lemma statement holds.
        \item We have that $e(s)=k$. Then $|f(s)|<r$, otherwise $v^\star$ is a leaf in $T_{n,k}$. Furthermore, we must have that $\ell=0$.

        In this case, $v^\star$ has two children labeled with $\tlabel{s0}$ and $\tlabel{s1}$, respectively. By induction, the games at the vertices labeled with $\tlabel{s0}$ and $\tlabel{s1}$, respectively, are won by $i_w$ if the games at the vertices labeled with $p_{0}(f(s0))$ and $p_{0}(f(s1))$ in $T_n$, respectively, are won by $i_w$ when $\mathcal{T}'$ is conducted without manipulations. 
        Since the game at $v^\star$ is not manipulated, we have that the game at vertex labeled with $p_{0}(f(s))$ in $T_n$ is won by $i_w$ when $\mathcal{T}'$ is conducted without manipulations, then $i_w$ wins the game at $v^\star$. Hence, the lemma statement holds.
        
        \item We have that $|f(s)|< r$, $e(s)<k$, and $\ell=k-e(s)$.

        In this case, $v^\star$ has three children labeled with $\tlabel{s0}$, $\tlabel{s1}$, and $\tlabel{s\bot}$, respectively.

        Consider the case where the game at $v^\star$ is manipulated. Then $\ell-1$ games on each tournament run from a leaf to each of the vertices labeled with $\tlabel{s0}$, $\tlabel{s1}$, and $\tlabel{s\bot}$, respectively, are manipulated.
        By induction, the games at the vertices labeled with $\tlabel{s0}$ and $\tlabel{s1}$, respectively, are won by $i_w$ if the games at the vertices labeled with $p_{1}(f(s0))=p_{0}(f(s))$ and $p_{1}(f(s1))=p_{0}(f(s))
        $ in $T_n$, respectively, are won by $i_w$ when $\mathcal{T}'$ is conducted without manipulations. 
        The game at the vertex labeled with $\tlabel{s\bot}$  is won by $i_w$ if the game at the vertex labeled with $p_{0}(f(s\bot))=p_{1}(f(s))$ in $T_n$ is won by $i_w$ when $\mathcal{T}'$ is conducted without manipulations. 
        Note that if the game at the vertex labeled with $p_{0}(f(s))$ in $T_n$ is won by $i_w$ when $\mathcal{T}'$ is conducted without manipulations, then $i_w$ also wins the game at the vertex labeled with $p_{1}(f(s))$ in $T_n$, since the latter vertex is the parent of the former, and $i_w$ is the strongest player.
        Hence, we can conclude that the game at $v^\star$ is won by $i_w$ if the game at the vertex labeled with $p_{0}(f(s))$ in $T_n$ is won by $i_w$ when $\mathcal{T}'$ is conducted without manipulations, and the lemma statement holds.
        %Hence, we have that the winner of $v^\star$ is at least as strong as the winner of the vertex labeled with $p_{0}(f(s))$ in $T$. 
        
        Consider the case where the game at $v^\star$ is not manipulated. 
        Then~$\ell$ games on each tournament run from a leaf to the vertex labeled with $\tlabel{s0}$, $\tlabel{s1}$, and $\tlabel{s\bot}$, respectively, are manipulated.
        By induction, the games at the vertices labeled with $\tlabel{s0}$ and $\tlabel{s1}$, respectively, are won by $i_w$ if the games at the vertices labeled with $p_{0}(f(s0))$ and $p_{0}(f(s1))$ in $T_n$, respectively, are won by $i_w$ when $\mathcal{T}'$ is conducted without manipulations.
For the game at the vertex labeled with $\tlabel{s\bot}$ we cannot apply the induction hypothesis, since $e(s)=e(s\bot)-1$ and hence $k-e(s\bot)-\ell=-1$. We only know that it is won by a player that is seeded into the subtree of $T_{n,k}$ that is rooted at the vertex labeled with $\tlabel{s\bot}$.
Note that if the game at the vertex labeled with $p_{0}(f(s0))$ or $p_{0}(f(s1))$ in $T_n$ is won by $i_w$ when $\mathcal{T}'$ is conducted without manipulations, then $i_w$ also wins the game at the vertex labeled with $p_{0}(f(s))$ in $T_n$, since the latter vertex is the parent of the former, and $i_w$ is the strongest player.
Since the game at $v^\star$ is not manipulated and $i_w$ is the strongest player, we have that the game at vertex labeled with $p_{0}(f(s))$ in $T_n$ is won by $i_w$ when $\mathcal{T}'$ is conducted without manipulations, then $i_w$ wins the game at $v^\star$. Hence, the lemma statement holds.
%However, note that the winner of the vertex labeled with $f(s)$ in $T$ is the stronger of the winners of the vertices labeled with $f(s0)$ and $f(s1)$ in $T$, respectively. It follows that the stronger player of the players that are the winners of the vertices labeled with $\tlabel{s0}$ and $\tlabel{s1}$, respectively, is at least as strong as the winner of the vertex labeled with $f(s)$ in $T$. If the player coming from the leaf labeled with $\tlabel{s\bot}$ wins the game at $v^\star$, then this player is even stronger. Hence, the statement holds.

        \item We have that $|f(s)|< r$, $e(s)<k$, and $\ell<k-e(s)$. 

In this case, $v^\star$ has three children labeled with $\tlabel{s0}$, $\tlabel{s1}$, and $\tlabel{s\bot}$, respectively.

        Consider the case where the game at $v^\star$ is manipulated. Then $\ell-1$ games on each tournament run from a leaf to the vertex labeled with $\tlabel{s0}$, $\tlabel{s1}$, and $\tlabel{s\bot}$, respectively, are manipulated.
        By induction, the games at the vertices labeled with $\tlabel{s0}$ and $\tlabel{s1}$, respectively, are won by $i_w$ if the games at the vertices labeled with $p_{k-e(s0)-\ell+1}(f(s0))=p_{k-e(s)-\ell}(f(s))$ and $p_{k-e(s1)-\ell+1}(f(s1))=p_{k-e(s)-\ell}(f(s))
        $ in $T_n$, respectively, are won by $i_w$ when $\mathcal{T}'$ is conducted without manipulations. 
        The game at the vertex labeled with $\tlabel{s\bot}$  is won by $i_w$ if the game at the vertex labeled with $p_{k-e(s\bot)-\ell+1}(f(s\bot))=p_{k-e(s)-\ell+1}(f(s))$ in $T_n$ is won by $i_w$ when $\mathcal{T}'$ is conducted without manipulations.  
        Note that if the game at the vertex labeled with $p_{k-e(s)-\ell}(f(s))$ in $T_n$ is won by $i_w$ when $\mathcal{T}'$ is conducted without manipulations, then $i_w$ also wins the game at the vertex labeled with $p_{k-e(s)-\ell+1}(f(s))$ in $T_n$, since the latter vertex is the parent of the former, and $i_w$ is the strongest player. 
        Hence, we can conclude that the game at $v^\star$ is won by $i_w$ if the game at the vertex labeled with $p_{k-e(s)-\ell}(f(s))$ in $T_n$ is won by $i_w$ when $\mathcal{T}'$ is conducted without manipulations, and the lemma statement holds.

        Consider the case where the game at $v^\star$ is not manipulated. 
        Then~$\ell$ games on each tournament run from a leaf to the vertex labeled with $\tlabel{s0}$, $\tlabel{s1}$, and $\tlabel{s\bot}$, respectively, are manipulated.
        By induction, the games at the vertex labeled with $\tlabel{s0}$ and $\tlabel{s1}$, respectively, are won by $i_w$ if the games at the vertices labeled with $p_{k-e(s0)-\ell}(f(s0))$ and $p_{k-e(s1)-\ell}(f(s1))$ in $T_n$, respectively, are won by $i_w$ when $\mathcal{T}'$ is conducted without manipulations.
        The game at the vertex labeled with $\tlabel{s\bot}$ is won by $i_w$ if the game at the vertex labeled with $p_{k-e(s\bot)-\ell}(f(s\bot))=p_{k-e(s)-\ell}(f(s))$ in $T_n$ is won by $i_w$ when $\mathcal{T}'$ is conducted without manipulations.
        Since the game at $v^\star$ is not manipulated and $i_w$ is the strongest player, we have that the game at vertex labeled with $p_{k-e(s)-\ell}(f(s))$ in $T_n$ is won by $i_w$ when $\mathcal{T}'$ is conducted without manipulations, then $i_w$ wins the game at $v^\star$. Hence, the lemma statement holds.
        %It follows that $v^\star$ has a winner that is at least as strong as the winner of the vertex labeled with $p_{k-e(s)-\ell}(f(s))$ in $T$.
    \end{itemize}
    This finishes the proof.
\end{proof}

\cref{thm:ternary} now immediately follows from \cref{lem:ternarycorrect}, \cref{obs:timeternary}, and \cref{lem:ternaryheight}.

\subsection{The Binary Case}

%\paragraph{Description of the robust tournament tree.}
In the following, we describe recursively how to obtain a (binary) tournament tree $T_{n,k}$ for a tournament with $n$ players. Here, $k$ denotes the number of manipulations that may occur in each tournament run. We will prove the following.

\begin{theorem}\label{thm:binary}
Given integers $n$ and $k$, a tournament tree $T_{n,k}$ of height at most $\lceil 1.44\log n \rceil + 3k$ can be computed in $O(|T_{n,k}|)$ time such that the following holds.
Let $N$ be a set of $n$ players, let $w$ be a winner function that exhibits a strongest player $i_w\in N$, and let $\sigma$ be a seeding for the players to the seed position names of $T_{n,k}$.
    If on each tournament run in $\mathcal{T}=(N,T_{n,k},w,\sigma)$ we have that at most $k$ games are manipulated, then $i_w$ wins the tournament.
\end{theorem}

Analogous to the ternary case, we label each vertex of $T_{n,k}$ with three parts:
\[
\tlabel{s},
\]
where $s$ is a binary string, $k$ is a non-negative integer, $f:\{0,1\}^*\rightarrow\{0,1\}^*$ is a function mapping binary strings to binary strings, and $e:\{0,1\}^*\rightarrow\mathbb{N}$ is a function mapping binary strings to non-negative integers. The functions $f$ and $e$ are defined as follows.
\begin{itemize}
    \item The function $f$ takes as input a string $s$, exhaustively remove the prefix 11 from $s$, and then exhaustively removes all substrings 011 from $s$.
    \item The function $e$ takes as input a string and counts how many times the function $f$ removes a prefix 11 or a substring of the form 011 from $s$.
    %, more formally
    %\[
    %e(s)=\frac{|s|-|f(s)|}{3}.
    %\]
\end{itemize}

Intuitively, the first part of the label corresponds to the transmission sequence as it is sent by Alice to Bob. The second part corresponds to the sequence that is reconstructed by Bob and believed by him to be the unmanipulated string. The third part counts how many manipulations were already identified by Bob.

We label the root of the tournament tree $T_{n,k}$ with label $\tlabel{\varepsilon}$, where $\varepsilon$ is the empty string. The label simplifies to $\tttlabel{\varepsilon}{\varepsilon}{0}$. 
Let $r$ denote the length of the seed position names. Note that since we will only use bit strings that do not contain consecutive ones as seed position names, using names of length $r$ does not give $2^r$ different names. We will explain later to which value we need to set $r$ to guarantee that we get at least $n$ different seed position names.
A vertex with label $\tlabel{s}$ such that $|f(s)|\le r$ and $e(s) \le k$ has up to two children.
\begin{itemize}
    \item If $|f(s)|< r$, then the vertex has one child with label $\tlabel{s0}$.
    \item If $|f(s)|< r$ or $e(s)< k$, then the vertex has one child with label $\tlabel{s1}$.
\end{itemize}
Here $s0$ and $s1$ are the strings $s$ with a $0$ and a $1$ appended, respectively. We give an illustration in \cref{fig:binary1}. Note that a vertex with label $\tlabel{s}$ such that $|f(s)| = r$ and $e(s)=k$ is a leaf. Then we consider the bit string $f(s)$ to be the seed position name associated with the leaf. If any inner vertex in the constructed tree has only one child, then we can copy this child (together with the subtree rooted at it) to obtain a full binary tree.
 See \cref{fig:binary2} for an illustration of the root of the tournament tree together with the three generations of children. The gray vertex in \cref{fig:binary2} is the first occurrence of a vertex label where the function $f$ removes an occurrence of the string 011 from $s$.
\begin{figure}[t]
\begin{center}

\includegraphics[scale=1]{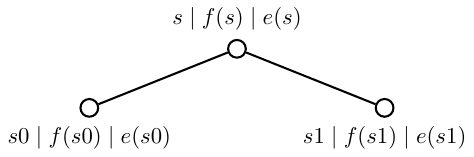}
%\begin{tikzpicture}[line width=1pt,scale=1,xscale=2.5]
%
%    \node[vert,label=above:$\tlabel{s}$] (v) at (0,0) {};
%    \node[vert,label=below:$\tlabel{s0}$] (v1) at (-1,-1) {};
%    \node[vert,label=below:$\tlabel{s1}$] (v2) at (1,-1) {};
%      
%    \draw (v) -- (v1);
%    \draw (v) -- (v2);
%\end{tikzpicture}
    \end{center}
    \caption{Visualization of the recursive construction of $T_{n,k}$ in the binary case.}\label{fig:binary1}
\end{figure}

\begin{figure}[t]
\noindent\makebox[\textwidth]{
\centering
%\begin{center}

\includegraphics[scale=1]{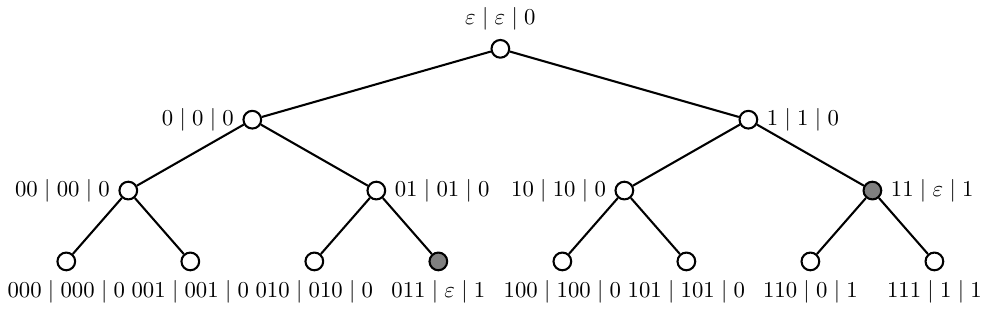}
%\begin{tikzpicture}[line width=1pt,yscale=1.2,xscale=2.1]
%
%    \node[vert,label=above:$\varepsilon\mid\varepsilon\mid 0$] (v0) at (2,1) {};
%
%    \node[vert,label=left:$0\mid0\mid 0$] (v) at (0,0) {};
%    \node[vert,label=left:$00\mid00\mid 0$] (v1) at (-1,-1) {};
%    \node[vert,label=right:$01\mid01\mid 0$] (v2) at (1,-1) {};
%    \node[vert,label=below:$000\mid000\mid 0$] (v21) at (-1.5,-2) {};
%    \node[vert,label=below:$001\mid001\mid 0$] (v22) at (-.5,-2) {};
%    \node[vert,label=below:$010\mid010\mid 0$] (v23) at (.5,-2) {};
%    \node[vert,fill=gray,label=below:$011\mid\varepsilon\mid 1$] (v24) at (1.5,-2) {};
%
%    \node[vert,label=right:$1\mid1\mid 0$] (w) at (4,0) {};
%    \node[vert,label=left:$10\mid10\mid 0$] (w1) at (3,-1) {};
%    \node[vert, fill=gray,label=right:$11\mid\varepsilon\mid 1$] (w2) at (5,-1) {};
%    \node[vert,label=below:$100\mid100\mid 0$] (w21) at (2.5,-2) {};
%    \node[vert,label=below:$101\mid101\mid 0$] (w22) at (3.5,-2) {};
%    \node[vert,label=below:$110\mid0\mid 1$] (w23) at (4.5,-2) {};
%    \node[vert,label=below:$111\mid1\mid 1$] (w24) at (5.5,-2) {};
%
%    \draw (v0) -- (v);
%    \draw (v0) -- (w);
%    
%    \draw (v) -- (v1);
%    \draw (v) -- (v2);
%    \draw (v1) -- (v21);
%    \draw (v1) -- (v22);
%    \draw (v2) -- (v23);
%    \draw (v2) -- (v24);
%      
%    \draw (w) -- (w1);
%    \draw (w) -- (w2);
%    \draw (w1) -- (w21);
%    \draw (w1) -- (w22);
%    \draw (w2) -- (w23);
%    \draw (w2) -- (w24);
%\end{tikzpicture}
 %   \end{center}
    }
    \caption{Illustration of the tournament tree $T_{n,k}$ in the binary case. Gray vertices visualize vertices with labels where the function $f$ removes occurrences of 001 from $s$.}\label{fig:binary2}
\end{figure}

% It is easy to see that all such strings can be characterized by walks in the graph depicted in \cref{fig:binary3} that start at the gray vertex. Hence, we need to compute the minimum walk length $x$ such that there are at least $n$ different walks in the graph depicted in \cref{fig:binary3} that start at the gray vertex. The adjacency matrix of the graph depicted in \cref{fig:binary3} is the following:
% \[
% A=
% \begin{pmatrix}
%     1 & 1\\
%     1 & 0
% \end{pmatrix},
% \]
% where the first row corresponds to the left (gray) vertex in \cref{fig:binary3}
% %, the middle row to the middle vertex in \cref{fig:binary3}, 
% and the last row to the right vertex in \cref{fig:binary3}. It is well-known that the number of walks of length $\ell$ from a vertex $i$ to a vertex $j$ is $(A^\ell)_{ij}$. It follows that the number of walks of length $\ell$ starting at the gray vertex in \cref{fig:binary3} is
% \[
% (A^\ell)_{11}+(A^\ell)_{12}.
% \]
% Hence, we have that 
% \[
% x=\min\{ \ell\mid (A^\ell)_{11}+(A^\ell)_{12}\ge n\}.
% \]
% \begin{figure}[t]
%     \centering
% \begin{tikzpicture}[line width=1pt,yscale=2,xscale=4]

%     %\node[vert,fill=gray] (v1) at (0,0) {};
%     \node[vert,fill=gray] (v2) at (1,0) {};
%     \node[vert] (v3) at (2,0) {};

%     %\draw[diredge] (v1) edge[loop above]node{1} (v1);
%     %\draw[diredge] (v1) --node[above]{0} (v2);
%     \draw[diredge] (v2) edge[loop above]node{0} (v2);
%     \draw[diredge] (v2) edge[bend left]node[above]{1} (v3);
%     \draw[diredge] (v3) edge[bend left]node[below]{0} (v2);
% \end{tikzpicture}
%     \caption{}\label{fig:binary3}
% \end{figure}}}

Note that given the label $\tlabel{s}$ of a vertex, we can compute the labels of the children in constant time. To do this, we do not recompute the functions $f$ and $e$ but rather compute the values of $f$ and $e$ in the labels of the children from the values in the label $\tlabel{s}$ in a straightforward manner. This allows us to observe the following.

\begin{observation}\label{obs:timebinary}
    The tournament tree $T_{n,k}$ can be computed in $O(|T_{n,k}|)$ time.
\end{observation}

Now we analyze the height of $T_{n,k}$.

\begin{lemma}\label{lem:binaryheight}
    The height of $T_{n,k}$ is at most $r + 3k$.
\end{lemma}
\begin{proof}
    Let $v^\star$ be a vertex in $T_{n,k}$ that is labeled with $\tlabel{s}$. We associate with $\tlabel{s}$ the value $r-|f(s)|+3(k-e(s))$. By construction, the vertex $v^\star$ in $T_{n,k}$ has up to two children, labeled $\tlabel{s0}$ and $\tlabel{s1}$, respectively. We have that $\tlabel{s0}$ is associated with the value $r-|f(s0)|+3(k-e(s0))$. Note that we have $e(s)=e(s0)$ and $|f(s)|=|f(s0)|-1$ and hence 
    \[
    r-|f(s0)|+3(k-e(s0))=r-|f(s)|+3(k-e(s))-1.
    \]
    %Note that we cannot have $e(s)<e(s0)$, since appending a 0 to a string cannot create a new occurrence of 011.
    The situation with $\tlabel{s1}$ if $s$ ends with a zero and $e(s)=e(s1)$ is analogous. However, if $s$ ends with a one we have that $e(s)=e(s1)-1$. In this case, we have that $|f(s)|=|f(s1)|+1$ or $|f(s)|=|f(s1)|+2$. It follows that we have 
    \[
    r-|f(s1)|+3(k-e(s1))\le r-|f(s)|+3(k-e(s))-1.
    \]
 Overall, we have that each child of a vertex is associated with a value that is by at least one smaller than that of the vertex itself. Clearly, the value cannot become negative. The root of $T_{n,k}$ is associated with value $r + 3k$. The lemma statement follows.
\end{proof}

%\paragraph{Transforming a tournament to a robust one.}
%Assume we are given a player set $N$ with $|N|=n$, a tournament seeding $\sigma$, and an integer $k$. Let $T$ be the ordered rooted complete binary tree with $|N|$ leaves. Assume the leaves of $T$ are labeled with binary strings of length $r$ that do not start with 11 or contain 011 as a substring. For an example see \cref{fig:labeling2}.
We want $r$ to be the smallest integer such that there are at least $n$ different binary strings of length $r$ that do not start with 11 or contain 011 as a substring. It is easy to see that those are the binary strings that do not contain consecutive ones. This will guarantee that we have sufficiently many different seed position names for $n$ players.

It is folklore that the number of binary strings of length $r$ that do not contain consecutive ones is $F_{r+2}$, where $F_r$ is the $r$th Fibonacci number.\footnote{To generate all strings of length $r$ without consecutive ones, we can do the following. Take all strings of length $r-1$ without consecutive ones and append a zero. This gives all strings of length $r$ without consecutive ones that end with zero. Take all strings without consecutive ones of length $r-2$ and append 01. This gives all strings of length $r$ without consecutive ones that end with one. The base case starts with 1 string of length zero and two strings of length 1.}
Furthermore, it is well-known that $F_{r+2}\ge \phi^r$, where $\phi=(1+\sqrt{5})/2$ is the golden ratio~\cite{gardner1959scientific}. Hence, we have that $F_{r+2}\ge n$ for all $r\ge \log n/\log \phi$, in particular for  $r = \lceil 1.44\log n \rceil$.
From \cref{lem:binaryheight} and the fact that we set $r = \lceil 1.44\log n \rceil$ we get the following.

\begin{corollary}\label{cor:binaryheight}
    The height of the robust tournament tree $T_{n,k}$ is at most $\lceil 1.44\log n \rceil + 3k$.
\end{corollary}

%Given integers $n$ and $k$, we construct a generalized tournament tree $T_{n,k}$ as described above.
%We identify player $i$ with the binary string $s_i$ of length $r$ that is the label of the leaf $\sigma^{-1}(i)$ of $T$. We define a generalized tournament seeding $\sigma'$ as follows. We seed player $i$ into every leaf of $T_{r,k}$ that is labeled with $\tlabel{s}$ such that $f(s)=s_i$. It is straightforward to see that each player is seeded into at least one leaf of $T_{r,k}$. Note that this does not necessarily assign a player to each leaf of $T_{r,k}$. Hence, we introduce an additional dummy player $d$ that is weaker than any other player in $N$, and assign all remaining leaves to this player. We update the DAG $D$ to contain $d$ as a new sink and call the new DAG $D'$. From \cref{lem:binaryheight} and the fact that we set $r = \lceil 1.44\log n \rceil$ we get the following.
%\begin{corollary}\label{cor:binaryheight}
 %   The height of the robust tournament tree is at most $\lceil 1.44\log n \rceil + 3k$.
%\end{corollary}

\begin{figure}[t]
%\noindent\makebox[\textwidth]{
\centering
%\begin{center}

\includegraphics[scale=1]{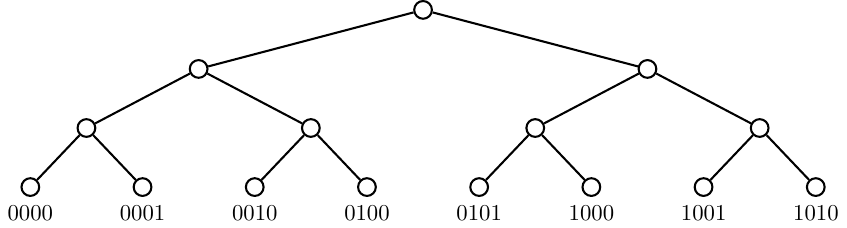}
%\begin{tikzpicture}[line width=1pt,yscale=1,xscale=1.9]
%
%    \node[vert] (v0) at (2,1) {};
%
%    \node[vert] (v) at (0,0) {};
%    \node[vert] (v1) at (-1,-1) {};
%    \node[vert] (v2) at (1,-1) {};
%    \node[vert,label=below:$0000$] (v21) at (-1.5,-2) {};
%    \node[vert,label=below:$0001$] (v22) at (-.5,-2) {};
%    \node[vert,label=below:$0010$] (v23) at (.5,-2) {};
%    \node[vert,label=below:$0100$] (v24) at (1.5,-2) {};
%
%    \node[vert] (w) at (4,0) {};
%    \node[vert] (w1) at (3,-1) {};
%    \node[vert] (w2) at (5,-1) {};
%    \node[vert,label=below:$0101$] (w21) at (2.5,-2) {};
%    \node[vert,label=below:$1000$] (w22) at (3.5,-2) {};
%    \node[vert,label=below:$1001$] (w23) at (4.5,-2) {};
%    \node[vert,label=below:$1010$] (w24) at (5.5,-2) {};
%
%    \draw (v0) -- (v);
%    \draw (v0) -- (w);
%    
%    \draw (v) -- (v1);
%    \draw (v) -- (v2);
%    \draw (v1) -- (v21);
%    \draw (v1) -- (v22);
%    \draw (v2) -- (v23);
%    \draw (v2) -- (v24);
%      
%    \draw (w) -- (w1);
%    \draw (w) -- (w2);
%    \draw (w1) -- (w21);
%    \draw (w1) -- (w22);
%    \draw (w2) -- (w23);
%    \draw (w2) -- (w24);
%\end{tikzpicture}
 %   \end{center}
%    }
    \caption{Illustration of a tournament tree that uses seed position names without consecutive ones.}\label{fig:labeling2}
\end{figure}

\begin{figure}[t]
\noindent\makebox[\textwidth]{
\centering

\includegraphics[scale=1]{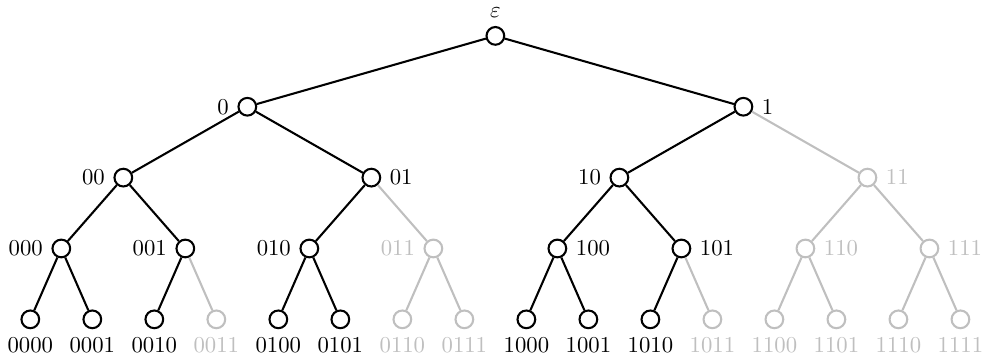}
    }
    \caption{Illustration of the prefix tree for the seed position names used in \cref{fig:labeling2}.}\label{fig:prefix}
\end{figure}

To show that $T_{n,k}$ is robust, we provide an analogous proof to the one for \cref{lem:ternarycorrect}. However, instead of considering whether the strongest player wins a certain game in a ``normal'' tournament with the same seed position names, we consider whether the seed position name of the strongest player has a certain prefix. In the previous section, those two ideas coincide, since the ``normal'' tournament tree $T_n$ (see \cref{fig:labeling}) is also a prefix tree for the seed position names. Here, this is not the case, since seed position names do not have consecutive ones (see \cref{fig:labeling2}). The prefix tree for the seed position names in \cref{fig:labeling2} is illustrated in \cref{fig:prefix}.
For the formal proof, we introduce some additional notation. 
\begin{definition}
    Let $s$ be a binary and let $\ell$ be some integer. We denote with $p_\ell(s)$ the string $f(s1^{\ell'})$ where $\ell'$ is the smallest integer such that $e(s1^{\ell'})=e(s)+\ell$ or $f(s1^{\ell'})=\varepsilon$. 
\end{definition}

In other words, $p_\ell(s)$ is the prefix obtained from $s$ by removing the number of bits that correspond to detecting $\ell$ additional errors.
We can observe the following.
\begin{observation}\label{obs:p}
  Let $s$ be a binary string and let $\ell$ be some integer.  
  
  If $f(s)$ ends with 1, then we have that $p_\ell(f(s))=p_{\ell-1}(f(s1))$ and $p_\ell(f(s))=p_{\ell+1}(f(s0))$.
  
  If $f(s)$ ends with 0, then we have that $p_\ell(f(s))=p_{\ell}(f(s1))$ and $p_\ell(f(s))=p_{\ell+1}(f(s0))$.
\end{observation}
\begin{proof}
This follows straightforwardly from the definition of $p_\ell$. Consider the case where $f(s)$ ends with a 1. Then adding another 1 to $s$ will cause the string $s$ to have an additional two consecutive ones, and hence $e(s)=e(s1)-1$. Let $\ell'$ be the smallest integer such that $e(s1^{\ell'})=e(s)+\ell$ or $f(s1^{\ell'})=\varepsilon$. Observe that we also have $e((s1)1^{\ell'-1})=e(s1)+\ell-1$ or $f((s1)1^{\ell'-1})=\varepsilon$. Furthermore, we must have that $\ell'-1$ is the smallest integer such that $e((s1)1^{\ell'-1})=e(s1)+\ell-1$ or $f((s1)1^{\ell'-1})=\varepsilon$. Otherwise, we get a contradiction to the assumption that $\ell'$ is the smallest integer such that $e(s1^{\ell'})=e(s)+\ell$ or $f(s1^{\ell'})=\varepsilon$. It follows that $p_\ell(f(s))=p_{\ell-1}(f(s1))$. By analogous considerations, the remaining three cases can be shown.
\end{proof}

Using this, we prove that the following invariant holds during the conduction of $\mathcal{T}$.
\begin{lemma}\label{lem:binarycorrect}
    Let $N$ be a set of $n$ players, let $w$ be a winner function that exhibits a strongest player $i_w\in N$, and let $\sigma$ be a seeding for the players to the seed position names of $T_{n,k}$.
    Let $s$ be a string on the alphabet $\{0,1\}$ such that $|f(s)|\le r$ and $e(s)\le k$, and let $0\le \ell\le k-e(s)$. If on each tournament run in $\mathcal{T}=(N,T_{n,k},w,\sigma)$ from a leaf of $T_{n,k}$ to the vertex $v^\star$ labeled with $\tlabel{s}$ we have that at most $\ell$ games are manipulated, then the following holds.
    If the seed position name $\sigma(i_w)$ assigned  to $i_w$ has $p_{k-e(s)-\ell}(f(s))$ as a prefix, then $i_w$ is the winner of $v^\star$.
\end{lemma}
\begin{proof}
    We prove the lemma by induction on the depth of $T_{n,k}$ (from the leaves to the root). Let $v$ by a leaf in $T_{n,k}$. By definition, $v$ is labeled with some $\tlabel{s}$ such that $e(s)=k$. Hence, we have $\ell=0$. By definition, we also have that the seed position name of leaf $v$ is $f(s)$. Since $p_0(f(s))=f(s)$, the statement holds.

    Consider a vertex $v^\star$ labeled with $\tlabel{s}$ such that $|f(s)|\le r$ and $e(s)\le k$ in $T_{n,k}$ that is not a leaf. Assume that on each tournament run from a leaf to this vertex, we have that at most $\ell\le k-e(s)$ games are manipulated. Consider the following cases.
    \begin{itemize}
        \item We have that $|f(s)|= r$. Then $e(s)<k$, otherwise $v^\star$ is a leaf in $T_{n,k}$. 
        
        By definition, $v^\star$ has one vertex as a child labeled with $\tlabel{s1}$. Note that $e(s)=e(s1)-1$ if $f(s)$ ends with a 1, and $e(s)=e(s1)$ otherwise. 
        
        If the game at $v^\star$ is not manipulated, then by induction, the game at the child of~$v^\star$ is won by $i_w$ if $\sigma(i_w)$ has $p_{k-e(s1)-\ell}(f(s1))$ as a prefix. 
        If the string $f(s)$ ends with a 1, then we have that $p_{k-e(s1)-\ell}(f(s1))=p_{k-e(s)-\ell}(f(s))$.
        If the string $f(s)$ ends with a 0, then by \cref{obs:p} we also have that $p_{k-e(s1)-\ell}(f(s1))=p_{k-e(s)-\ell}(f(s1))=p_{k-e(s)-\ell}(f(s))$.
        
        If the game at $v^\star$ is manipulated, then by induction, the game at the child of~$v^\star$ is won by $i_w$ if $\sigma(i_w)$ has $p_{k-e(s1)-\ell+1}(f(s1))$ as a prefix.
        If the string $f(s)$ ends with a 1, then we have that $p_{k-e(s1)-\ell+1}(f(s1))=p_{k-e(s)-\ell+1}(f(s))$.
        If the string $f(s)$ ends with a 0, then by \cref{obs:p} we also have that $p_{k-e(s1)-\ell+1}(f(s1))=p_{k-e(s)-\ell+1}(f(s1))=p_{k-e(s)-\ell+1}(f(s))$. 

        Hence, in both above cases, we can conclude that the game at $v^\star$ is won by $i_w$ if $\sigma(i_w)$ has $p_{k-e(s)-\ell}(f(s))$ as a prefix and the lemma statement holds.
        \item We have that $e(s)=k$. Then $|f(s)|<r$, otherwise $v^\star$ is a leaf in $T_{n,k}$. Furthermore, we must have that $\ell=0$.

        In this case, $v^\star$ has two children labeled with $\tlabel{s0}$ and $\tlabel{s1}$, respectively. By induction, the games at the vertices labeled with $\tlabel{s0}$ and $\tlabel{s1}$, respectively, are won by $i_w$ if $\sigma(i_w)$ has $p_{0}(f(s0))$ and $p_{0}(f(s1))$, respectively, as a prefix. We have $p_{0}(f(s0))=f(s0)$ and $p_{0}(f(s1))=f(s1)$. Since the game at $v^\star$ is not manipulated, we have that the game is won by $i_w$ if $\sigma(i_w)$ has $p_{0}(f(s))=f(s)$ as a prefix. Hence, the lemma statement holds.
        
        \item We have that $|f(s)|< r$, $e(s)<k$, and $\ell=k-e(s)$.

        In this case, $v^\star$ has two children labeled with $\tlabel{s0}$ and $\tlabel{s1}$, respectively. Note that $e(s)=e(s0)$, and note that $e(s)=e(s1)-1$ if $f(s)$ ends with a 1, and $e(s)=e(s1)$ otherwise.

        Consider the case where the game at $v^\star$ is manipulated. Then $\ell-1$ games on each tournament run from a leaf to each of the vertices labeled with $\tlabel{s0}$ and $\tlabel{s1}$, respectively, are manipulated. 
        By induction, the games at the vertices labeled with $\tlabel{s0}$ and $\tlabel{s1}$, respectively, are won by $i_w$ if $\sigma(i_w)$ has $p_{1}(f(s0))$ and $p_{x}(f(s1))$, respectively, as a prefix, where $x=0$ if the string $f(s)$ ends with a 1 and $x=1$ otherwise. By \cref{obs:p} we have that $p_{1}(f(s0))=p_0(f(s))$. Furthermore, by \cref{obs:p} we have that if the string $f(s)$ ends with a 1, then $p_{0}(f(s1))=p_1(f(s))$, and if the string $f(s)$ ends with a 0, then $p_{1}(f(s1))=p_1(f(s))$. 
        Note that $p_1(f(s))$ is a prefix of $p_0(f(s))$.
        Hence, we have that if $\sigma(i_w)$ has $p_{0}(f(s))=f(s)$, then $i_w$ wins the games at both leafs of $v^\star$. It follows that the winner of the game at $v^\star$ is $i_w$ if $\sigma(i_w)$ has $p_{0}(f(s))=f(s)$ as a prefix and the lemma statement holds.

        Consider the case where the game at $v^\star$ is not manipulated. 
        Then~$\ell$ games on each tournament run from a leaf to the vertex labeled with $\tlabel{s0}$ and $\tlabel{s1}$, respectively, are manipulated. We now make a case distinction on whether the string $f(s)$ ends with a 0 or a 1. First, consider the case where the string $f(s)$ ends with a 0.
        Then by induction, the games at the vertices labeled with $\tlabel{s0}$ and $\tlabel{s1}$, respectively, are won by $i_w$ if $\sigma(i_w)$ has $p_{0}(f(s0))$ and $p_{0}(f(s1))$, respectively, as a prefix. By \cref{obs:p} we have that $p_{0}(f(s1))=p_0(f(s))$. 
        Hence, if the player assigned to the leaf labeled with $\tlabel{s1}$ wins the game at $v^\star$, then the game is won by $i_w$ if $\sigma(i_w)$ has $p_{0}(f(s))$ as a prefix. If the player assigned to the leaf labeled with $\tlabel{s0}$ wins the game at $v^\star$ and this player is different from $i_w$, then we have that $\sigma(i_w)$ does not have $p_0(f(s))$ as a prefix. This is true because otherwise, that is, if $\sigma(i_w)$ has $p_{0}(f(s))$ as a prefix, then we have established that $i_w$ wins the game at the leaf of $v^\star$ that is labeled with 
        $\tlabel{s1}$, and since the game at $v^\star$ is not manipulated and $i_w$ is the strongest player, $i_w$ wins the game at $v^\star$.
        Now consider the case where the string $f(s)$ ends with a 1.
        Then by induction, the game at the vertex labeled with $\tlabel{s0}$ is won by $i_w$ if $\sigma(i_w)$ has $p_{0}(f(s0))$ as a prefix. 
        For the game at the vertex labeled with $\tlabel{s1}$ we cannot apply the induction hypothesis, since $e(s)=e(s1)-1$ and hence $k-e(s1)-\ell=-1$. We only know that it is won by a player that is seeded into a position in the subtree of $T_{n,k}$ that is rooted at the vertex labeled with $\tlabel{s1}$. However, since the string $f(s)$ ends with a 1, we have that each seed position name that has $f(s)$ as a prefix also has $f(s)0=f(s0)$ as a prefix, because the seed position names cannot have two consecutive ones. 
        It follows that if $\sigma(i_w)$ has $p_{0}(f(s))$ as a prefix, it also has $p_{0}(f(s0))$ as a prefix and hence wins the game at the leaf of $v^\star$ labeled with $\tlabel{s0}$. Since the game at $v^\star$ is not manipulated and $i_w$ is the strongest player, the game at $v^\star$ is won by $i_w$ if $\sigma(i_w)$ has $p_{0}(f(s))=f(s)$ as a prefix and the lemma statement holds.
        %
        %It follows that if the player from the leaf labeled with $\tlabel{s0}$ wins the game at $v^\star$, then the game is won by a player that is at least as strong as the strongest player among all players whose names have $p_{0}(f(s))$ as a prefix. 
        %If the player from the leaf labeled with $\tlabel{s1}$ wins the game at $v^\star$, then this player is even stronger than the one from the leaf labeled with $\tlabel{s0}$. Hence, the statement holds.

        \item We have that $|f(s)|< r$, $e(s)<k$, and $\ell<k-e(s)$. 

In this case, $v^\star$ has two children labeled with $\tlabel{s0}$ and $\tlabel{s1}$, respectively. Note that $e(s)=e(s0)$, and note that $e(s)=e(s1)-1$ if $f(s)$ ends with a 1, and $e(s)=e(s1)$ otherwise.

        Consider the case where the game at $v^\star$ is manipulated. Then $\ell-1$ games on each tournament run from a leaf to the vertex labeled with $\tlabel{s0}$ and $\tlabel{s1}$, respectively, are manipulated.
By induction, the games at the vertices labeled with $\tlabel{s0}$ and $\tlabel{s1}$, respectively, are won by $i_w$ if $\sigma(i_w)$ has $p_{k-e(s0)-\ell+1}(f(s0))$ and $p_{k-e(s1)-\ell+1}(f(s1))$, respectively, as a prefix. By \cref{obs:p} we have that $p_{k-e(s0)-\ell+1}(f(s0))=p_{k-e(s)-\ell}(f(s))$. Furthermore, by \cref{obs:p} we have that if the string $f(s)$ ends with a 1, then $p_{k-e(s1)-\ell+1}(f(s1))=p_{k-e(s)-\ell+1}(f(s))$, and if the string $f(s)$ ends with a 0, then $p_{k-e(s1)-\ell+1}(f(s1))=p_{k-e(s)-\ell}(f(s))$. 
Note that $p_{k-e(s)-\ell+1}(f(s))$ is a prefix of $p_{k-e(s)-\ell}(f(s))$.
Hence, we have that if $\sigma(i_w)$ has $p_{k-e(s)-\ell}(f(s))$ as a prefix, then $i_w$ wins the games at both leafs of $v^\star$. It follows that the winner of the game at $v^\star$ is $i_w$ if $\sigma(i_w)$ has $p_{k-e(s)-\ell}(f(s))$ as a prefix and the lemma statement holds.
%Hence, we have that the winner of $v^\star$ is at least as strong as the strongest player among all players whose names have $p_{k-e(s)-\ell}(f(s))$ as a prefix and the lemma statement holds.

        Consider the case where the game at $v^\star$ is not manipulated. 
        Then~$\ell$ games on each tournament run from a leaf to the vertex labeled with $\tlabel{s0}$ and $\tlabel{s1}$, respectively, are manipulated.
By induction, the games at the vertices labeled with $\tlabel{s0}$ and $\tlabel{s1}$, respectively, are won by $i_w$ if $\sigma(i_w)$ has $p_{k-e(s0)-\ell}(f(s0))$ and $p_{k-e(s1)-\ell}(f(s1))$, respectively, as a prefix.
By \cref{obs:p} we have that if the string $f(s)$ ends with a 1, then $p_{k-e(s1)-\ell}(f(s1))=p_{k-e(s1)-\ell+1}(f(s))=p_{k-e(s)-\ell}(f(s))$, and if the string $f(s)$ ends with a 0, then $p_{k-e(s1)-\ell}(f(s1))=p_{k-e(s1)-\ell}(f(s))=p_{k-e(s)-\ell}(f(s))$.
It follows that if $\sigma(i_w)$ has $p_{k-e(s)-\ell}(f(s))$ as a prefix, then $i_w$ wins the game at the leaf of $v^\star$ labeled with $\tlabel{s1}$. Since the game at $v^\star$ is not manipulated and $i_w$ is the strongest player, the game at $v^\star$ is won by $i_w$ if $\sigma(i_w)$ has $p_{k-e(s)-\ell}(f(s))=f(s)$ as a prefix and the lemma statement holds.
%It follows that if the player from the leaf labeled with $\tlabel{s1}$ wins the game at $v^\star$, then the game is won by a player that is at least as strong as the strongest player among all players whose names have $p_{k-e(s)-\ell}(f(s))$ as a prefix. If the player from the leaf labeled with $\tlabel{s0}$ wins the game at $v^\star$, then this player is even stronger than the one from the leaf labeled with $\tlabel{s1}$.
    \end{itemize}
    This finishes the proof.
\end{proof}

\cref{thm:binary} now immediately follows from \cref{lem:binarycorrect}, \cref{obs:timebinary}, and \cref{cor:binaryheight}.
Finally, we discuss how to obtain different bounds on the depth of $T_{n,k}$. 
We can modify the communication protocol in \cref{sec:protocol2} to use longer sequences of consecutive ones to indicate errors in the communication.
Accordingly, we can modify the functions $f$ and $e$, that is, $f$ exhaustively removes the prefix $1^x$ from the input string and then removes all occurrences of the bit string $01^x$ for some fixed $x$ from the input string. The function $e$ is defined analogously. Our analysis in \cref{lem:binaryheight} is done for $x=2$. 
If we choose a larger $x$, then we effectively increase the number of different seed position names of length $\ell$. Informally speaking, we can trade off the linear factor in front of $k$ with the linear factor in front of $\log n$ in the height of $T_{n,k}$. For example, for $x=3$ we can estimate the number of different bit strings of length $\ell$ using the so-called tribonacci numbers, a generalization of the Fibonacci numbers~\cite{gardner1959scientific}.
For the $r$th tribonacci number it is known it roughly equals $1.84^r$~\cite{gardner1959scientific}. This yields an upper bound for the height of $T_{n,k}$ of $\lceil1.14 \log n\rceil+4k$.
The formal argument is an analogous argument to the $x=2$ case (and a similar argument also holds for $x\ge 4$). We can conclude the following.
\begin{corollary}
    There exists a function $g$ such that every $\varepsilon>0$ there is a robust tournament tree $T_{n,k}$ of height at most $\lceil (1+\varepsilon)\log n\rceil + g(\varepsilon^{-1})\cdot k$.
\end{corollary}
This means that for constant $k$, we can get the height of the robust tournament tree to be arbitrarily close to $\log n$.

\section{Conclusion}
We showed how to obtain tournament trees that are robust to manipulation by adding redundancy. 
In particular, we show that it is possible to withstand up to a $1/3$ fraction of manipulations along each leaf-to-root path, with the trade-off being only a polynomial increase in the tournament size.
To this end, we reveal a surprising relation between robust tournament design and communication protocols that use a feedback channel and are robust against adversarial noise. With our work, we lay down the foundation for future research in this direction. 
There are many natural further questions.
\begin{itemize}
    \item Can we obtain lower bounds on the height of robust tournament trees?
    \item Can we obtain tournament trees that work in settings where there is no strongest player?\footnote{It is relatively easy to see that the answer to this question is yes if we allow the height of the tournament tree to be linear in $k\cdot\log n$. Hence, the question is whether we can obtain a height that is linear in $\log n+k$.}
    \item Can we extend the techniques for other tournament formats or other collective decision making mechanisms?
\end{itemize}

\bibliographystyle{abbrvnat}
\bibliography{bib}	

\end{document}